\documentclass[letterpaper,journal]{IEEEtran}
\usepackage{amsmath}
\usepackage{amssymb}
\usepackage{times}
\usepackage{graphicx}
\usepackage[colorlinks=true, allcolors=blue]{hyperref}
\usepackage{booktabs}
\usepackage{array}
\usepackage{makecell}
\usepackage{pifont}
\usepackage[T1]{fontenc}
\DeclareFontShape{T1}{ptm}{m}{scit}{<-> ssub * ptm/m/sc}{}
\usepackage{multirow}
\usepackage{longtable}    
\usepackage{tabularx}
\usepackage{microtype}  
\usepackage[table]{xcolor}
\usepackage[linesnumbered,ruled,vlined]{algorithm2e}

\usepackage{amsthm}
\usepackage[english]{babel}
\theoremstyle{plain}
\newtheorem{theorem}{Theorem}
\newtheorem{proposition}[theorem]{Proposition}
\newtheorem{lemma}[theorem]{Lemma}
\newtheorem{corollary}[theorem]{Corollary}
\theoremstyle{definition}
\newtheorem{definition}[theorem]{Definition}
\newtheorem{assumption}[theorem]{Assumption}
\theoremstyle{remark}

\usepackage{tikz}
\usetikzlibrary{arrows.meta,calc,fit,positioning,shapes.geometric}
\tikzset{
  paperbox/.style={draw=black!55, rounded corners=2pt, align=center, inner sep=4pt, fill=white},
  paperarrow/.style={-{Latex[length=2mm]}, semithick, draw=black!75, rounded corners=2pt},
  papersoft/.style={draw=black!40, rounded corners=2pt, align=center, inner sep=3.5pt, fill=black!4},
  paperaccent/.style={draw=black!75, rounded corners=2pt, align=center, inner sep=4pt, fill=black!8},
  papergreen/.style={draw=black!75, rounded corners=2pt, align=center, inner sep=4pt, fill=black!16},
  paperorange/.style={draw=black!75, rounded corners=2pt, align=center, inner sep=4pt, fill=black!24}
}

\UseRawInputEncoding

\setcellgapes{2pt}
\makegapedcells
\AtBeginDocument{%
  \fontdimen2\font=0.22em 
  \fontdimen3\font=0.1em  
  \fontdimen4\font=0.06em  
}

\begin{document}

\title{Prediction-Robust Service Deployment with Capacity-Aware Edge Admission}

\author{Hailiang Zhao, Ziqi Wang, Yifei Zhang, Mingyi Liu, Xinkui Zhao, Kingsum Chow, and Shuiguang Deng%
\thanks{H. Zhao, Z. Wang, Y. Zhang, X. Zhao, and K. Chow are with the School of Software Technology, Zhejiang University, Ningbo 315048, China (e-mail: hliangzhao@zju.edu.cn; wangziqi0312@zju.edu.cn; zhangyifei@zju.edu.cn; zhaoxinkui@zju.edu.cn; kingsum.chow@zju.edu.cn).}%
\thanks{M. Liu is with the Faculty of Computing, Harbin Institute of Technology, Harbin 150001, China (e-mail: lmy@hit.edu.cn).}%
\thanks{S. Deng is with the College of Computer Science and Technology, Zhejiang University, Hangzhou 310027, China (e-mail: dengsg@zju.edu.cn).}%
}



\maketitle

\begin{abstract}
Edge platforms instantiate executable services close to users to reduce request-serving cost, but each instance incurs a one-time deployment cost and remains useful only for a finite time-to-live (TTL). The resulting online decision is both prediction-sensitive and capacity-coupled: an optimistic forecast can waste deployment cost, whereas a delayed decision misses the burst it is intended to serve. We study this problem under a common TTL cost model and propose \textsc{CAPSUM}, a capacity-aware admission policy with an elastic specialization, \textsc{CAPSUM-E}. In the local elastic setting, every node-service trace is exactly a variable-price Bahncard instance. This reduction lets \textsc{CAPSUM-E} inherit PFSUM's tight prediction-error-dependent ratio, including $2/(1+\beta)$ consistency and $1/\beta$ robustness for $\beta>0$. A redirect-aware variant preserves the same local deployment schedule. For finite-capacity nodes, \textsc{CAPSUM} combines size-scaled break-even tests, a utilization-dependent shadow price, and evidence-density eviction; we prove capacity feasibility, scale invariance, and exact agreement with \textsc{CAPSUM-E} under an elastic configuration. We implement an exact local offline dynamic program and compare against direct common-model baselines and documented source-derived adapters for EDP-A, OREO, and uEDC-L. Experiments cover controlled prediction error, three synthetic demand regimes, a causal predictor on a public Globus Compute trace, and joint scaling to 1,024 nodes and 10,000 services. Under the common model, \textsc{CAPSUM} reduces normalized cost by 33.7--42.9\% relative to the best source-derived adapter across the synthetic regimes and by 45.5\% on the sampled trace.
\end{abstract}

\begin{IEEEkeywords}
Edge service deployment, edge computing, learning-augmented algorithms, online algorithms, and competitive analysis.
\end{IEEEkeywords}

\section{Introduction}
\label{sec:intro}

Edge platforms increasingly host stateless microservices and serverless functions close to users. A nearby instance can avoid a wide-area round trip and reduce the per-request resource cost, whereas a cold instance must first fetch an image, initialize its runtime, and reserve execution state. These effects make edge deployment a rent-or-buy decision: paying an up-front instantiation cost is worthwhile only if the requests served during the instance's useful lifetime recover that cost. Contemporary systems make this decision amid geographic demand shifts, short-lived hotspots, and heterogeneous service footprints. Consequently, a placement chosen from a global average popularity can be economically wrong at an individual edge node even when the global forecast is accurate.

Figure~\ref{fig:motivation} illustrates the setting. A crowd of AR users may create a local burst at $e_1$, and then move towards $e_2$ before an instance at $e_1$ expires. Sending all requests to the origin avoids an unnecessary deployment but repeatedly incurs the origin-serving cost during the burst. In contrast, deploying whenever a predictor signals demand can pay for an instance after a false alarm or just as the demand leaves the node. Deploying only after observing the burst is also inadequate: the past workload cannot be served retroactively by the newly created instance. The decision must therefore be local in space, forward looking in time, and safe when the forecast is wrong.

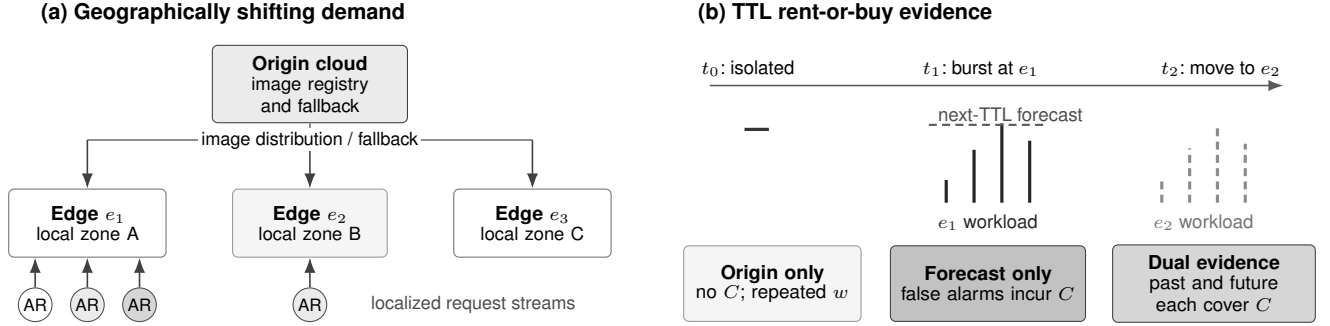
\begin{figure*}[t]
\centering
\begin{tikzpicture}[x=.92cm,y=1cm,font=\sffamily\fontsize{7.2}{8.4}\selectfont]
  \node[font=\sffamily\fontsize{8.2}{9.4}\selectfont\bfseries,anchor=west] at (0,3.95) {(a) Geographically shifting demand};
  \node[paperaccent, text width=23mm, minimum height=9mm] (cloud) at (4.0,3.02)
    {\textbf{Origin cloud}\\[-1pt]{\sffamily\fontsize{6.8}{7.8}\selectfont image registry and fallback}};
  \node[paperbox, text width=18mm, minimum height=9mm] (e1) at (0.8,1.18)
    {\textbf{Edge $e_1$}\\[-1pt]{\sffamily\fontsize{6.8}{7.8}\selectfont local zone A}};
  \node[papersoft, text width=18mm, minimum height=9mm] (e2) at (4.0,1.18)
    {\textbf{Edge $e_2$}\\[-1pt]{\sffamily\fontsize{6.8}{7.8}\selectfont local zone B}};
  \node[paperbox, text width=18mm, minimum height=9mm] (e3) at (7.2,1.18)
    {\textbf{Edge $e_3$}\\[-1pt]{\sffamily\fontsize{6.8}{7.8}\selectfont local zone C}};
  \draw[semithick,draw=black!55] (0.8,2.28) -- (7.2,2.28);
  \draw[semithick,draw=black!55] (cloud.south) -- (4.0,2.28);
  \draw[paperarrow] (0.8,2.28) -- (e1.north);
  \draw[paperarrow] (4.0,2.28) -- (e2.north);
  \draw[paperarrow] (7.2,2.28) -- (e3.north);
  \node[font=\sffamily\fontsize{6.8}{7.8}\selectfont,fill=white,inner sep=1pt] at (4.0,2.28)
    {image distribution / fallback};
  \node[circle,draw=black!70,fill=white,minimum size=4.2mm,inner sep=1pt,font=\sffamily\fontsize{6.4}{7.4}\selectfont] (u11) at (0.05,0.10) {AR};
  \node[circle,draw=black!70,fill=black!8,minimum size=4.2mm,inner sep=1pt,font=\sffamily\fontsize{6.4}{7.4}\selectfont] (u12) at (0.80,0.10) {AR};
  \node[circle,draw=black!70,fill=black!16,minimum size=4.2mm,inner sep=1pt,font=\sffamily\fontsize{6.4}{7.4}\selectfont] (u13) at (1.55,0.10) {AR};
  \draw[paperarrow] (u11.north) -- (u11.north |- e1.south);
  \draw[paperarrow] (u12.north) -- (e1.south);
  \draw[paperarrow] (u13.north) -- (u13.north |- e1.south);
  \node[circle,draw=black!70,fill=black!8,minimum size=4.2mm,inner sep=1pt,font=\sffamily\fontsize{6.4}{7.4}\selectfont] (u21) at (4.0,0.10) {AR};
  \draw[paperarrow] (u21.north) -- (e2.south);
  \node[font=\sffamily\fontsize{6.8}{7.8}\selectfont,align=center,text=black!70] at (6.35,0.10)
    {localized request streams};

  \node[font=\sffamily\fontsize{8.2}{9.4}\selectfont\bfseries,anchor=west] at (9.45,3.95) {(b) TTL rent-or-buy evidence};
  \draw[-{Latex[length=2mm]},semithick,black!65] (9.75,3.00) -- (18.00,3.00);
  \node[above,font=\sffamily\scriptsize] at (10.30,3.00) {$t_0$: isolated};
  \node[above,font=\sffamily\scriptsize] at (13.65,3.00) {$t_1$: burst at $e_1$};
  \node[above,font=\sffamily\scriptsize] at (17.10,3.00) {$t_2$: move to $e_2$};
  \draw[black!80,very thick] (10.25,2.42) -- ++(0.35,0);
  \foreach \x/\h in {13.15/0.30,13.55/0.70,13.95/1.05,14.35/0.82} {
    \draw[black!85,very thick] (\x,1.45) -- ++(0,\h);
  }
  \foreach \x/\h in {16.25/0.30,16.65/0.72,17.05/1.02,17.45/0.78} {
    \draw[black!45,very thick,densely dashed] (\x,1.45) -- ++(0,\h);
  }
  \node[below,font=\sffamily\scriptsize,text=black!85] at (13.75,1.42) {$e_1$ workload};
  \node[below,font=\sffamily\scriptsize,text=black!55] at (16.85,1.42) {$e_2$ workload};
  \draw[densely dashed,black!65,thick] (12.90,2.48) -- (14.55,2.48);
  \node[font=\sffamily\scriptsize,anchor=west,text=black!70] at (12.90,2.62) {next-TTL forecast};
  \node[papersoft, text width=21mm, minimum height=10mm] at (10.65,0.38)
    {\textbf{Origin only}\\[-1pt]\scriptsize no $C$; repeated $w$};
  \node[paperorange, text width=23mm, minimum height=10mm] at (13.75,0.38)
    {\textbf{Forecast only}\\[-1pt]\scriptsize false alarms incur $C$};
  \node[papergreen, text width=24mm, minimum height=10mm] at (17.00,0.38)
    {\textbf{Dual evidence}\\[-1pt]\scriptsize past and future each cover $C$};
\end{tikzpicture}
\caption{Motivating scenario for decentralized TTL-based edge service deployment. Requests are spatially localized and may move between edge zones before an instantiated service expires. Origin-only service forgoes savings during a burst; forecast-only deployment can incur cost $C$ after a false alarm or demand migration; and a history-only policy deploys too late to serve a completed burst. The decision is whether an edge should instantiate a service for TTL $T$ from its local stream and a short-horizon forecast while retaining robustness to prediction error.}
\label{fig:motivation}
\end{figure*}

The literature has made substantial progress on joint service placement, task scheduling, and resource allocation. These formulations typically optimize a coupled, capacity-limited system through centralized optimization or learned policies~\cite{afachao2024efficient,du2025online,wang2025joint}. They address important orchestration objectives but do not isolate the online decision that precedes, or accompanies, capacity admission: given an inactive service at one edge, should the node create a $T$-valid instance despite uncertain future demand? Recent prediction-assisted service-caching and offloading policies likewise exploit time-varying demand~\cite{yang2025service}, but prediction accuracy alone provides no worst-case safeguard. An adversarially optimistic forecast can induce a deployment whose realized workload is insufficient to amortize its cost.

We first isolate an uncapacitated deployment layer. Its decision object is an executable service instance with a prescribed time-to-live (TTL), rather than a cache item whose paid lifetime may be terminated by an unrelated admission. This scope captures serverless warm starts, image prefetching, and elastic edge storage while separating prediction robustness from cache admission, CPU scheduling, and dependency-aware placement. The resulting local decision has an exact economic break-even point: only workload arriving after deployment and before expiry can amortize the instantiation cost. Consequently, aggregating past and predicted workloads into one statistic loses the relevant temporal asymmetry. Section~\ref{sec:capacity} then introduces finite capacity and explicitly permits early eviction without refunding the sunk instantiation cost.

Learning-augmented online algorithms are designed precisely for this setting: they can benefit from accurate advice while retaining an explicit bound under arbitrary advice~\cite{lykouris2018competitive,bamas2020primal}. PFSUM, developed for the variable-price Bahncard problem, combines recent observations with a short-horizon workload prediction~\cite{zhao2024learning}. We show that its two guards have a direct systems interpretation. The historical guard filters a single unsupported optimistic forecast; the future guard tests whether an instance bought now can still generate sufficient savings during its own TTL. Neither guard alone supplies both properties.

The edge setting raises two further questions. First, can every node apply this principle using only local history and local forecasts while retaining a system-wide guarantee? Second, what remains true when a request can be served by a live instance at another edge node? We answer the first through an exact decomposition: under local-only serving and the uncapacitated assumption, every node-service stream is an independent Bahncard instance. We answer the second with a redirect-aware extension whose deployment schedule is intentionally unchanged from the analyzable local policy; redirection is used only when that policy would otherwise forward a request to the origin. This separation is essential because an arbitrary coordinator that suppresses a locally justified deployment changes the online policy and cannot inherit the local competitive proof automatically.

Our contributions are as follows.
\begin{enumerate}
    \item We formulate decentralized TTL service deployment with explicit instantiation, lifetime, request-serving, prediction-error, and redirection semantics, and separate the analyzable elastic layer from finite-capacity admission.
    \item We develop \textsc{CAPSUM-E} as a decentralized safety layer and prove an exact PFSUM reduction, local trace equivalence, and a tight $\kappa(\eta)$ bound. We further derive an offline-cost-weighted aggregate error certificate. Its redirect-aware extension preserves the deployment schedule and is $\kappa(\eta)/\beta$-competitive against the redirectable offline benchmark for $\beta>0$.
    \item We extend the method to finite-capacity nodes as \textsc{CAPSUM}. Its size-scaled dual evidence, local shadow price, and evidence-density eviction make every action capacity feasible; we prove scale invariance and exact agreement with \textsc{CAPSUM-E} in the elastic configuration. The finite-capacity policy is not claimed to inherit the PFSUM ratio.
    \item We build a reproducible simulator with an exact local offline dynamic program, direct common-model baselines, and documented source-derived adapters for EDP-A, OREO, and uEDC-L. The 895-run evaluation mirrors the claims: exact-offline and prediction-error studies test the elastic mechanism, whereas ablation, synthetic-regime, causal-trace, and joint-scale studies test finite-capacity admission.
\end{enumerate}

The remainder of this paper is organized as follows. Section~\ref{sec:related} reviews related work. Section~\ref{sec:problem} defines the deployment model. Section~\ref{sec:algorithm-analysis} presents \textsc{CAPSUM-E} and its local-model analysis. Section~\ref{sec:capacity} gives the capacity-aware extension, and Section~\ref{sec:redirectable} considers request redirection. Section~\ref{sec:evaluation} reports the reproducible evaluation, and Section~\ref{sec:conclusion} concludes.

\section{Related Work}
\label{sec:related}
The closest research spans edge caching, executable-service placement, and
learning-augmented online algorithms. These areas share terminology but not
decision semantics: the analyzable instance is committed for a fixed TTL and
cannot discount prior requests, whereas finite-capacity cache entries may be
evicted early and therefore require a separate model.

\subsection{Edge Caching and Service Placement}

Collaborative data caching jointly selects storage locations and retrieval
paths~\cite{xia2021online}; context-aware D2D caching adds online learning under
storage and energy budgets~\cite{xia2024context}; and LRU-BaSE distinguishes
byte- and object-miss objectives for variable-size content~\cite{wang2024belady}.
EDP-A places erasure-coded blocks near requesting users while preserving
recoverability~\cite{luo2025placement}. Mean-field and DRL policies address
large caching populations and adaptive content demand~\cite{feng2023meanfield,wang2024ice}.
TMC service-caching formulations additionally model shared VM, computation, and
bandwidth resources~\cite{xu2023collaborative}, Bayesian pricing
\cite{tutuncuoglu2024optimal}, cold-start-aware retention~\cite{xiao2024cold},
and popularity or failure uncertainty~\cite{luo2025cost}. These objectives are
capacity coupled; hit rate or forecast accuracy alone does not characterize
size-scaled instantiation cost. We therefore compare EDP-A and uEDC-L only in
the finite-capacity track using the audited adapters in Section~\ref{sec:evaluation}.

Executable-service systems optimize a broader orchestration problem. TPDS work
studies distributed serverless DAG embedding, FaaS cost-performance selection,
and SLO-aware layer sharing~\cite{deng2022dependent,lin2023faas,cheng2024slo}.
TSC formulations cover FaaS provisioning~\cite{ascigil2022resource}, cross-edge
container caching and routing~\cite{chen2024cross}, hierarchical online MEC
control~\cite{du2025online}, joint caching/offloading and resource allocation
\cite{tang2025collaborative}, and separate data and service placement
\cite{wang2025joint}. Redundant placement addresses service reliability
\cite{zhao2020distributed}, while OREO combines Lyapunov optimization and Gibbs
sampling for joint caching and offloading~\cite{xu2018joint}. Learning-based
placement can exploit telemetry and graph structure~\cite{qiao2019deep,Hou2022Graph,goudarzi2021distributed},
but distributional performance does not imply pathwise safety under arbitrary
forecast error. Our elastic model instead isolates when an otherwise feasible
local instance should be created; hard resource coupling is handled separately
by \textsc{CAPSUM} and is not covered by the elastic theorem.

\subsection{Learning-Augmented Online Algorithms}

Learning-augmented algorithms quantify both the benefit of accurate advice and
safety under arbitrary advice~\cite{lykouris2018competitive,bamas2020primal},
including results for caching, metric task systems, and parsimonious prediction
access~\cite{antoniadis2020online,im2022parsimonious}. PFSUM treats continuous
time and heterogeneous prices in the variable-price Bahncard problem, buying a
$T$-valid discount only when recent and predicted future costs both reach the
break-even point~\cite{zhao2024learning}. We prove the exact node-service
mapping rather than assuming this result transfers to edge orchestration,
establish trajectory equivalence for decentralized execution, and preserve the
trajectory under redirection. The capacity-aware policy remains explicitly
outside this separable guarantee.

\section{Problem Formulation and System Model}
\label{sec:problem}

\subsection{System Architecture and Scope}

We consider a set of edge nodes $\mathcal{N}$, with $N:=|\mathcal{N}|$, and a catalog $\mathcal{S}$ of stateless, executable services. A service can represent a containerized microservice or a serverless function that must be instantiated before it can serve a request at an edge node. The origin server stores every service and is always available. Time is continuous.

\begin{assumption}[Uncapacitated deployment layer]
\label{asm:uncapacitated}
An edge node can host every service that it elects to deploy during the horizon of interest. Thus, a deployment of $(n,s)$ neither evicts nor changes the serving cost of another node-service pair.
\end{assumption}

This assumption deliberately isolates the online deployment decision from cache admission, CPU scheduling, and placement under hard capacity constraints. It is appropriate when the deployment layer is elastic or when these resources are provisioned separately. It is also the property that makes the local model separable; a finite-capacity extension is not covered by the guarantees below.

\subsection{Request and Cost Model}

A finite request sequence is $\sigma=\langle\sigma_1,\ldots,\sigma_L\rangle$, where $\sigma_i=(t_i,n_i,s_i,w_i)$ has time $t_i$, ingress node $n_i\in\mathcal{N}$, requested service $s_i\in\mathcal{S}$, and regular service cost $w_i>0$. Requests are revealed in strictly increasing time order. The regular cost aggregates the resource and network cost of serving the request from the origin.

Instantiating a service at an edge node costs $C>0$. If an instance of $s_i$ is active at $n_i$ when $\sigma_i$ arrives, the request is served locally at cost $\beta w_i$, where $0\leq\beta<1$. Otherwise the request is served by the origin at cost $w_i$. Hence $(1-\beta)w_i$ is the saving delivered by a local instance on request $i$. The analysis permits heterogeneous positive $w_i$; no bounded-request-cost assumption is required.

\subsection{Deployment State and Objective}

\begin{definition}[Deployment state]
\label{def:deployment}
A deployment action $(n,s,\tau)$ instantiates service $s$ at node $n$ at time $\tau$ and remains active on $[\tau,\tau+T)$, where $T>0$ is a fixed time-to-live (TTL). A service is \emph{active} at $(n,t)$ if one of the algorithm's deployment actions for $(n,s)$ covers $t$.
\end{definition}

Let $\mathcal{D}_{\mathcal A}(\sigma)$ be the deployment actions made by an online algorithm $\mathcal A$. It incurs
\begin{equation}
\label{eq:total-cost}
 \mathcal A(\sigma)
 =C|\mathcal{D}_{\mathcal A}(\sigma)|
 +\sum_{i=1}^{L}\mathrm{cost}_{\mathcal A}(\sigma_i),
\end{equation}
where
\begin{equation}
\label{eq:per-request-cost}
\mathrm{cost}_{\mathcal A}(\sigma_i)=
\begin{cases}
\beta w_i, & \text{if $s_i$ is active at $n_i$ at time $t_i$},\\
w_i, & \text{otherwise}.
\end{cases}
\end{equation}
A deployment may be made immediately before serving a request, so that request receives the discounted cost. Overlapping deployments of the same pair are never useful. Moreover, an offline optimum can, without loss of generality, deploy only immediately before requests for an inactive pair: delaying any other deployment until the next such request cannot increase its cost.

The local-service offline benchmark is
\begin{equation}
\label{eq:opt-cost}
\mathrm{OPT}^{\mathrm{loc}}(\sigma)
:=\min_{\mathcal D}
\left\{C|\mathcal D|+\sum_{i=1}^{L}
\mathrm{cost}_{\mathcal D}(\sigma_i)\right\},
\end{equation}
where $\mathcal D$ ranges over valid TTL deployment schedules and all inactive requests are served by the origin.

\subsection{Learning-Augmented Information}
\label{sec:la-framework}

For a node-service pair $(n,s)$, define its total regular workload in an interval $I$ as
\begin{equation}
\label{eq:true-workload}
 c_{n,s}(\sigma;I):=
 \sum_{\sigma_j\in\sigma:\,t_j\in I,\,n_j=n,\,s_j=s}w_j.
\end{equation}
Whenever an algorithm encounters a request for an inactive local pair $(n_i,s_i)$ at time $t_i$, it receives a nonnegative prediction
$\hat c_{n_i,s_i}(\sigma;[t_i,t_i+T))$ of the workload in the following TTL window. The predictor is a black box and may be different across nodes.

Let $\mathcal I_{\mathcal A}(\sigma)$ be the indices of requests that arrive when $(n_i,s_i)$ is inactive under $\mathcal A$. Write $\hat F_i:=\hat c_{n_i,s_i}(\sigma;[t_i,t_i+T))$ and $F_i:=c_{n_i,s_i}(\sigma;[t_i,t_i+T))$. The realized error of $\mathcal A$ is
\begin{equation}
\label{eq:prediction-error}
 \eta_{\mathcal A}(\sigma):=
 \max_{i\in\mathcal I_{\mathcal A}(\sigma)}
 \left|\hat F_i-F_i\right|,
\end{equation}
with $\max\varnothing:=0$. This is the same short-horizon, workload-valued error measure used by PFSUM~\cite{zhao2024learning}; it is not a count of mispredicted requests.

The economic break-even workload is
\begin{equation}
\label{eq:break-even}
 \gamma:=\frac{C}{1-\beta}.
\end{equation}
A deployment at time $t$ can recoup its purchase cost over $[t,t+T)$ exactly when the regular workload in that interval reaches $\gamma$.

\subsection{Competitive Criterion and Problem Statement}

For an algorithm $\mathcal A$, define
\begin{equation}
\label{eq:competitive-ratio}
 \mathrm{CR}_{\mathcal A}(\eta):=
 \sup_{\substack{\sigma:\,\eta_{\mathcal A}(\sigma)\leq\eta\\
                  \mathrm{OPT}^{\mathrm{loc}}(\sigma)>0}}
 \frac{\mathcal A(\sigma)}{\mathrm{OPT}^{\mathrm{loc}}(\sigma)}.
\end{equation}
An algorithm is $\delta$-consistent if $\mathrm{CR}_{\mathcal A}(0)\leq\delta$ and $\vartheta$-robust if $\mathrm{CR}_{\mathcal A}(\eta)\leq\vartheta$ for every $\eta\geq0$.

Our first objective is a decentralized policy that uses only the local history and local predictions of each $(n,s)$ pair and attains a tight consistency-robustness trade-off under the local-service model. We then consider a redirectable-service extension in which an inactive local request may be served by a live remote instance. The latter has a stronger offline benchmark and therefore requires a separate analysis.

\section{Algorithm Design and Theoretical Analysis}
\label{sec:algorithm-analysis}

\subsection{Elastic CAPSUM Core}
\label{subsec:algorithm-design}

\textsc{CAPSUM-E} runs the PFSUM purchasing rule independently at every node-service pair. On a request for an inactive local pair, it deploys only when both the workload observed in the preceding TTL window and the predicted workload in the following TTL window reach $\gamma$. The historical window includes the current request and all previous requests, whether or not they were served locally. This detail is required by the PFSUM analysis. Figure~\ref{fig:architecture} depicts the resulting node-local execution path and the redirect-aware overlay.

\begin{figure*}[t]
\centering
\begin{tikzpicture}[x=1cm,y=1cm,font=\sffamily\fontsize{7.2}{8.4}\selectfont]
  \node[papersoft, text width=22mm, minimum height=12mm] (request) at (1.05,1.65)
    {\textbf{Request at $n$}\\[-1pt]$(t,n,s,w)$};
  \node[paperbox, text width=25mm, minimum height=12mm] (state) at (4.10,1.65)
    {\textbf{Pair-local state}\\[-1pt]history queue; expiry};
  \node[paperbox, text width=25mm, minimum height=11mm] (predictor) at (4.10,3.10)
    {\textbf{Local predictor}\\[-1pt]$\hat c_{n,s}([t,t+T))$};
  \node[paperaccent, text width=30mm, minimum height=15mm] (decision) at (8.00,1.65)
    {\textbf{Dual-evidence gate}\\[-1pt]$H=c_{n,s}((t-T,t])$\\[-1pt]deploy iff $H,\hat F\geq\gamma$};
  \node[papergreen, text width=24mm, minimum height=11mm] (local) at (12.70,2.85)
    {\textbf{Local instance}\\[-1pt]start TTL; cost $\beta w$};
  \node[paperbox, text width=27mm, minimum height=11mm] (directory) at (12.05,0.38)
    {\textbf{Optional redirect overlay}\\[-1pt]query verified live instances};
  \node[papersoft, text width=22mm, minimum height=10mm] (remote) at (16.15,1.72)
    {\textbf{Remote edge}\\[-1pt]nearest live copy};
  \node[paperorange, text width=22mm, minimum height=10mm] (origin) at (16.15,0.38)
    {\textbf{Origin}\\[-1pt]fallback cost $w$};
  \draw[paperarrow] (request) -- (state);
  \draw[paperarrow] (state) -- (decision);
  \draw[paperarrow] (request.north) |- (predictor.west);
  \draw[paperarrow] (predictor.east) -| (decision.north);
  \coordinate (deployturn) at (10.15,2.85);
  \draw[paperarrow] (decision.east)
    -- (10.15,1.65) -- (deployturn)
    -- node[midway,above=2pt,fill=white,inner sep=1pt,font=\sffamily\fontsize{6.8}{7.8}\selectfont] {deploy} (local.west);
  \coordinate (nodeployturn) at (8.00,0.38);
  \draw[paperarrow] (decision.south) -- (nodeployturn)
    -- node[midway,above=2pt,fill=white,inner sep=1pt,font=\sffamily\fontsize{6.8}{7.8}\selectfont] {no deploy} (directory.west);
  \draw[paperarrow] (directory.north) -- (12.05,1.72)
    -- node[midway,above=2pt,fill=white,inner sep=1pt,font=\sffamily\fontsize{6.8}{7.8}\selectfont] {verified} (remote.west);
  \draw[paperarrow] (directory.east)
    -- node[midway,above=2pt,fill=white,inner sep=1pt,font=\sffamily\fontsize{6.8}{7.8}\selectfont] {fallback} (origin.west);
  \node[fit=(state)(predictor)(decision), draw=black!45, dashed, rounded corners=3pt, inner sep=4pt,
        label={[font=\sffamily\fontsize{6.8}{7.8}\selectfont]above:node-local safety path}] {};
\end{tikzpicture}
\caption{Decentralized execution architecture. \textsc{CAPSUM-E} needs only the local stream, local state, and a local forecast. \textsc{CAPSUM-E+} consults the directory only after the local dual-threshold rule elects not to deploy, so the analytically required local deployment schedule is preserved.}
\label{fig:architecture}
\end{figure*}
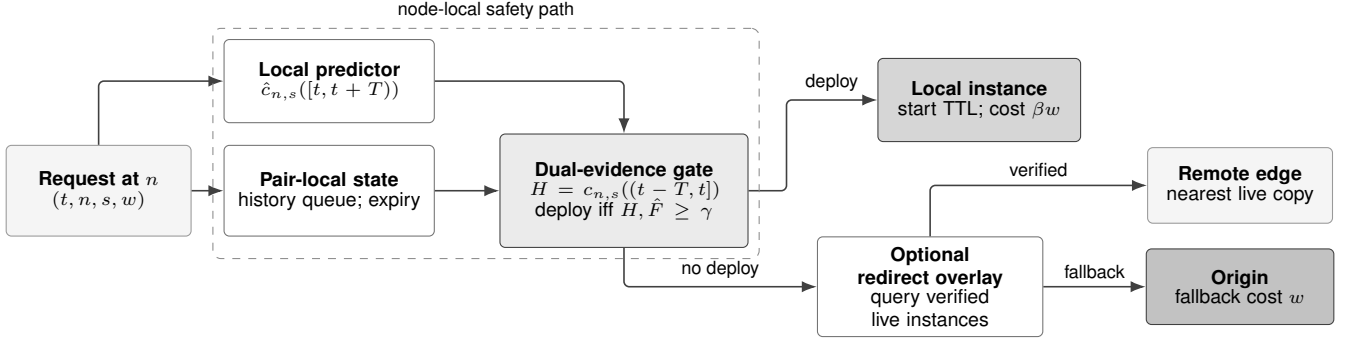

\begin{algorithm}[t]
\caption{\textsc{CAPSUM-E} at node $n_i$ upon request $\sigma_i$}
\label{alg:capsum-e}
\KwIn{$\sigma_i=(t_i,n_i,s_i,w_i)$; $C>0$, $0\leq\beta<1$, and $T>0$}
Update the local sliding-window history for $(n_i,s_i)$ with $\sigma_i$\;
$H\gets c_{n_i,s_i}(\sigma;(t_i-T,t_i])$\;
\uIf{$s_i$ is active at $n_i$}{
    Serve $\sigma_i$ locally at cost $\beta w_i$\;
}
\Else{
    Query $\hat F\gets\hat c_{n_i,s_i}(\sigma;[t_i,t_i+T))$\;
    \uIf{$H\geq\gamma$ \textbf{and} $\hat F\geq\gamma$}{
        Deploy $s_i$ at $n_i$ for TTL $T$ and pay $C$\;
        Serve $\sigma_i$ locally at cost $\beta w_i$\;
    }
    \Else{
        Forward $\sigma_i$ to the origin at cost $w_i$\;
    }
}
\end{algorithm}

The policy requires no inter-node communication. A node can maintain each recent workload with a timestamped queue and a running sum, giving $O(1)$ amortized history-update time per request; the prediction-query cost depends on the deployed predictor. The memory required by exact queues is proportional to the number of locally observed requests in the trailing window, and can be traded for time-bucketed counters if an approximation is acceptable.

\subsection{Design Rationale}
\label{subsec:design-rationale}

The two guards have distinct roles. The future guard asks whether an instance purchased now can plausibly amortize its cost during its own validity interval. The past guard prevents the policy from purchasing solely because of an arbitrary optimistic forecast. The windows are deliberately half-open: the current request belongs to both $H=c((t-T,t])$ and the true future workload $c([t,t+T))$, while a request exactly at $t+T$ belongs to neither active interval of an instance bought at $t$. This convention is identical to the TTL semantics in Definition~\ref{def:deployment}.

\begin{proposition}[Why the historical guard is necessary]
\label{prop:history-guard}
The forecast-only threshold rule that deploys whenever
$\hat c_{n,s}([t,t+T))\geq\gamma$ has no bounded competitive ratio under
arbitrary predictions.
\end{proposition}

\begin{proof}
Consider a single request of cost $\varepsilon>0$ at time $t$ and let the predictor return $\hat c([t,t+T))=\gamma$. The policy deploys and pays $C+\beta\varepsilon$, while $\mathrm{OPT}^{\mathrm{loc}}$ forwards the request and pays $\varepsilon$. Letting $\varepsilon$ tend to zero makes the ratio unbounded. In contrast, \textsc{CAPSUM-E} does not deploy because its historical workload is $H=\varepsilon<\gamma$.
\end{proof}

Past and future workloads are not summed because elapsed requests cannot benefit
from an instance bought now. If $W$ is the realized workload in $[t,t+T)$,
deployment changes the interval cost by $C-(1-\beta)W$ and breaks even exactly
at $W=\gamma$. The forecast estimates this unknown $W$; the historical guard
independently rejects unsupported optimism. A timestamped queue maintains the
half-open historical window in $O(1)$ amortized update time.

This separation also explains why a single weighted threshold does not, in
general, preserve the two certificates: the historical term can compensate for
insufficient post-deployment savings, or the forecast term can compensate for
an empty history. Requiring each quantity to cover the same economic threshold
preserves the temporal interpretation of both tests. At an
inactive decision point, the required state is therefore minimal: an
active-until timestamp, the exact sliding-window sum, and one pair-local
forecast. No node needs another node's history, forecast, or active set for the
elastic local-service rule.

\subsection{Exact Reduction to Independent Bahncard Instances}
\label{subsec:reduction}

For the reduction, request $(t_i,n,s,w_i)$ becomes a Bahncard ticket of price
$p_i=w_i$; deployment is a card of price $C$, validity $T$, and discount
$\beta$. The ticket price is $w_i$, not the saving $(1-\beta)w_i$, so prediction
and error retain the workload units of \eqref{eq:prediction-error}.

For every $(n,s)$, let $\sigma^{(n,s)}$ contain precisely the requests in $\sigma$ with ingress node $n$ and service $s$. Assumption~\ref{asm:uncapacitated} and local-only serving imply
\begin{align}
\textsc{CAPSUM-E}(\sigma)
 &=\sum_{n\in\mathcal N}\sum_{s\in\mathcal S}
   \textsc{CAPSUM-E}_{n,s}(\sigma^{(n,s)}), \label{eq:cost-decomp}\\
\mathrm{OPT}^{\mathrm{loc}}(\sigma)
 &=\sum_{n\in\mathcal N}\sum_{s\in\mathcal S}
   \mathrm{OPT}^{\mathrm{loc}}_{n,s}(\sigma^{(n,s)}). \label{eq:opt-decomp}
\end{align}

\begin{lemma}[PFSUM reduction]
\label{lem:variable-bahncard}
For a fixed pair $(n,s)$, identify every request $(t_i,n,s,w_i)$ with a Bahncard request $(t_i,p_i)$ of price $p_i=w_i$. Then the actions and costs of $\textsc{CAPSUM-E}_{n,s}$ are exactly those of PFSUM for $\mathrm{BP}(C,\beta,T)$.
\end{lemma}

\begin{proof}
A live local service instance is a valid Bahncard: it costs $C$, remains valid for $T$, and changes a request price from $w_i$ to $\beta w_i$. At every regular request, Algorithm~\ref{alg:capsum-e} tests $c(\sigma^{(n,s)};(t_i-T,t_i])\geq\gamma$ and the predicted $T$-future cost against the same $\gamma=C/(1-\beta)$, which is precisely the PFSUM rule. PFSUM is defined for continuous time and heterogeneous request prices, so no additional variable-weight extension is needed~\cite{zhao2024learning}.
\end{proof}

\begin{proposition}[Local trace equivalence]
\label{prop:trace-equivalence}
Fix a request sequence and the predictor response returned at each decision point. A centralized implementation that maintains a separate PFSUM state for every $(n,s)$ and processes the same requests produces exactly the same deployments, local/origin serving actions, and total cost as the decentralized execution of \textsc{CAPSUM-E}.
\end{proposition}

\begin{proof}
The state used to decide a request $(t_i,n_i,s_i,w_i)$ consists only of the active-until time, the sliding-window workload, and the predictor response for $(n_i,s_i)$. By Assumption~\ref{asm:uncapacitated}, an action for another pair changes none of these quantities. Induction over the strictly ordered requests therefore gives identical local states before every request and identical actions thereafter. Summing the identical per-request and deployment costs proves the claim.
\end{proof}

Thus decentralization exactly factorizes the centralized policy; the conclusion
fails once one pair changes another pair's feasibility or service cost.

For later use, let $\eta_{n,s}$ denote the maximum error in \eqref{eq:prediction-error} over decision points of the $(n,s)$ substream. The PFSUM result gives the following function:
\begin{equation}
\label{eq:kappa}
\kappa(\eta)=
\begin{cases}
\dfrac{2\gamma+(2-\beta)\eta}
      {(1+\beta)\gamma+\beta\eta}, & 0\leq\eta\leq\gamma,\\[8pt]
\dfrac{(3-\beta)\gamma+\eta}
      {(1+\beta)\gamma+\beta\eta}, & \eta>\gamma.
\end{cases}
\end{equation}

\begin{lemma}[Per-pair guarantee]
\label{lem:per-subproblem}
For every $(n,s)$,
\begin{equation}
\textsc{CAPSUM-E}_{n,s}(\sigma^{(n,s)})
\leq\kappa(\eta_{n,s})\,
\mathrm{OPT}^{\mathrm{loc}}_{n,s}(\sigma^{(n,s)}).
\label{eq:per-sub-guarantee}
\end{equation}
The bound is tight for a single pair.
\end{lemma}

\begin{proof}
Lemma~\ref{lem:variable-bahncard} maps the subproblem exactly to PFSUM. The stated expression and tightness follow from the PFSUM competitive-ratio theorem~\cite{zhao2024learning}.
\end{proof}

\subsection{Global Guarantee}
\label{subsec:main-theorem}

\begin{proposition}[Properties of $\kappa$]
\label{prop:kappa-properties}
The function $\kappa$ is continuous and strictly increasing on $[0,\infty)$. Moreover, $\kappa(0)=2/(1+\beta)$, and, for $\beta>0$, $\lim_{\eta\to\infty}\kappa(\eta)=1/\beta$.
\end{proposition}

\begin{proof}
The branches agree at $\eta=\gamma$; differentiation gives positive numerators
$(2+\beta)(1-\beta)\gamma$ and $(1-\beta)^2\gamma$. Substitution gives the
endpoints; for $\beta=0$, $\kappa$ is unbounded.
\end{proof}

\begin{theorem}[Local-model consistency and robustness]
\label{thm:main}
Let $\eta:=\eta_{\textsc{CAPSUM-E}}(\sigma)$. Then
\begin{equation}
\label{eq:main-cr}
\frac{\textsc{CAPSUM-E}(\sigma)}
     {\mathrm{OPT}^{\mathrm{loc}}(\sigma)}
\leq\kappa(\eta).
\end{equation}
The bound is tight. Thus, for $\beta>0$, \textsc{CAPSUM-E} is $2/(1+\beta)$-consistent and $1/\beta$-robust.
\end{theorem}

\begin{proof}
Since $\eta_{n,s}\leq\eta$, monotonicity, Lemma~\ref{lem:per-subproblem}, and
\eqref{eq:cost-decomp}--\eqref{eq:opt-decomp} give
\[
\textsc{CAPSUM-E}(\sigma)
\leq\sum_{n,s}\kappa(\eta_{n,s})\mathrm{OPT}^{\mathrm{loc}}_{n,s}
\leq\kappa(\eta)\mathrm{OPT}^{\mathrm{loc}}(\sigma).
\]
Tightness follows from a tight single-pair instance.
\end{proof}

\begin{corollary}[Instance-dependent aggregate certificate]
\label{cor:aggregate-certificate}
For a nonempty trace, define
$\lambda_{n,s}:=\mathrm{OPT}^{\mathrm{loc}}_{n,s}(\sigma^{(n,s)})/
\mathrm{OPT}^{\mathrm{loc}}(\sigma)$. Then
\begin{equation}
\label{eq:aggregate-certificate}
\frac{\textsc{CAPSUM-E}(\sigma)}{\mathrm{OPT}^{\mathrm{loc}}(\sigma)}
\leq\sum_{n\in\mathcal N}\sum_{s\in\mathcal S}
\lambda_{n,s}\kappa(\eta_{n,s})
\leq\kappa\!\left(\max_{n,s}\eta_{n,s}\right).
\end{equation}
\end{corollary}

\begin{proof}
Divide the per-pair sum by $\mathrm{OPT}^{\mathrm{loc}}(\sigma)$ and use
monotonicity with $\sum_{n,s}\lambda_{n,s}=1$.
\end{proof}

The corollary prevents an inaccurate but economically insignificant pair from
determining the trace-level certificate.

\subsection{Interpretation of the Guarantee and Its Scope}
\label{subsec:bound-interpretation}

Prediction quality enters only through the dimensionless error $\eta/\gamma$.
Scaling $C$, all $w_i$, and all predictions by the same positive factor scales
$\gamma$, online cost, and offline cost equally, leaving decisions and the ratio
unchanged. The TTL affects realized workload and forecast difficulty, not the
algebraic form of $\kappa$.

The parameter $\beta$ has a distinct systems interpretation. When $\beta$ is small, a local instance provides a large per-request saving, but missed deployments under poor forecasts may be far from optimal, as reflected by the $1/\beta$ robustness limit. As $\beta$ approaches one, local and origin service become nearly identical; the ratio approaches one because the deployment decision becomes economically degenerate, not because a larger $\beta$ represents a more efficient edge.

The result does not claim that a centralized controller cannot improve a practical metric. It states only that, under Assumption~\ref{asm:uncapacitated} and local-only serving, central knowledge cannot couple otherwise independent subproblems. When services compete for a finite resource budget, \eqref{eq:opt-decomp} no longer holds and neither the reduction nor the global theorem applies unchanged.

\section{Capacity-Aware Local Admission}
\label{sec:capacity}

The elastic specialization provides the analytical safety baseline, whereas production edges have finite image-storage, memory, or accelerator reservations. The full \textsc{CAPSUM} policy extends this core with local admission for the coupled setting. It retains the dual-evidence rationale while ensuring capacity feasibility, without extending the separable competitive theorem to a nonseparable model.

\subsection{Capacity Model and Shadow-Priced Evidence}

Let service $s$ consume an integer size $a_s\geq1$ and let node $n$ expose a finite budget $B_n$. At time $t$, its active set $\mathcal A_n(t)$ must obey
\begin{equation}
\label{eq:capacity-constraint}
 \sum_{s\in\mathcal A_n(t)}a_s\leq B_n.
\end{equation}
We charge a size-$a_s$ deployment $C_s=a_sC$, so its unpriced break-even workload is $a_s\gamma$. Let $u_n(t)$ be the left-hand side of \eqref{eq:capacity-constraint} and define the node-local scarcity price
\begin{equation}
\label{eq:shadow-price}
 \lambda_n(t)=\gamma\left(\frac{u_n(t)}{B_n}\right)^2.
\end{equation}
The square is a smooth, monotone engineering choice that keeps the price small at low utilization and raises it near saturation; it is not derived as a dual optimum. An operator may substitute another nonnegative local price curve without changing the feasibility argument below, although empirical behavior may change.

On a miss for $(n,s)$, \textsc{CAPSUM} computes the same historical workload $H$ and forecast $\hat F$ as Algorithm~\ref{alg:capsum-e}. It calls the candidate eligible only if
\begin{equation}
\label{eq:capsum-gate}
 H\geq a_s(\gamma+\lambda_n(t))\quad\text{and}\quad
 \hat F\geq a_s(\gamma+\lambda_n(t)).
\end{equation}
Thus a larger image must justify both its larger instantiation cost and the capacity it consumes. Its evidence density is $r_{n,s}=\min\{H,\hat F\}/a_s$. If there is insufficient free capacity, the node considers active residents in increasing order of their stored admission-evidence density and evicts only residents whose density is strictly below $r_{n,s}$, stopping once $a_s$ units are available. The stored density is the value recorded at the resident's admission and is not refreshed without another request for that pair. Early eviction terminates the remaining paid TTL and yields no refund. If the rule cannot free enough space, the request is forwarded and no deployment cost is paid. The rule uses only node-local state; it neither relies on a global cache map nor hides an infeasible admission.

\subsection{Safety Properties and Scope}

\begin{proposition}[Feasibility and elastic-limit compatibility]
\label{prop:capsum}
Assume $a_s\leq B_n$ for every service eligible at node $n$. \textsc{CAPSUM} maintains \eqref{eq:capacity-constraint} after every request. Define the elastic configuration by $a_s=1$, $B_n\geq|\mathcal S|$, $\lambda_n(t)=0$, and no eviction. In this configuration, \textsc{CAPSUM} makes exactly the same deployment decisions as \textsc{CAPSUM-E} for the same forecasts and request sequence.
\end{proposition}

\begin{proof}
An admission is made only if the candidate fits in the currently free capacity or if the selected evictions free at least $a_s$ units. All other actions preserve the active set, so induction over request arrivals establishes \eqref{eq:capacity-constraint}. In the stated elastic configuration, the scarcity price is zero and no capacity rejection or eviction can occur. Equation~\eqref{eq:capsum-gate} therefore becomes $H,\hat F\geq\gamma$, exactly the \textsc{CAPSUM-E} gate.
\end{proof}

\begin{proposition}[Scale invariance]
\label{prop:scale-invariance}
If $C$, all workloads, and all forecasts are scaled by any $\zeta>0$, with other inputs fixed, \textsc{CAPSUM-E}, \textsc{CAPSUM-E+}, and \textsc{CAPSUM} preserve every decision and scale all incurred costs by $\zeta$.
\end{proposition}

\begin{proof}
$\gamma$, $H$, $\hat F$, $\lambda_n(t)$, and all evidence densities scale by $\zeta$. Gates and density orderings are unchanged, hence so are active sets and routing; deployment and serving costs scale by $\zeta$.
\end{proof}

Propositions~\ref{prop:capsum} and~\ref{prop:scale-invariance} establish feasibility, elastic-limit reduction, and invariance to the workload-cost unit, not a capacity-competitive ratio. Capacity couples node-service pairs and invalidates the decomposition in \eqref{eq:opt-decomp}; a competitive guarantee for \textsc{CAPSUM} would require an explicitly capacity-aware offline benchmark and a separate analysis. Section~\ref{sec:evaluation} therefore evaluates \textsc{CAPSUM} only in the finite-capacity track and does not use its results to substantiate Theorem~\ref{thm:main}.

\section{Extension to the Redirectable-Service Model}
\label{sec:redirectable}

\subsection{Redirectable-Service Cost Model}
\label{subsec:redir-model}

Suppose that a request arriving at $n$ may be served by an active instance at another node $m$. Let $d$ be a metric on $\mathcal N$, and let $\alpha>0$ convert distance to per-unit-workload routing cost. The cost of assigning request $\sigma_i$ to $m_i\in\mathcal N\cup\{\bot\}$ is
\begin{equation}
\label{eq:redirect-cost}
\mathrm{cost}^{\mathrm{redir}}(\sigma_i)=
\begin{cases}
\beta w_i, & m_i=n_i,\\
(\beta+\alpha d(n_i,m_i))w_i,
  & m_i\in\mathcal N\setminus\{n_i\},\\
w_i, & m_i=\bot.
\end{cases}
\end{equation}
We assume $\alpha d(n,m)\leq1-\beta$ for all $n,m$, so a verified redirection is never more expensive than origin service. An assignment $m_i\in\mathcal N$ is feasible only if service $s_i$ is active at $m_i$ at time $t_i$ under the selected deployment schedule; $m_i=\bot$ denotes origin service. The redirectable offline benchmark jointly chooses valid deployments and feasible assignments:
\begin{equation}
\label{eq:opt-redir}
\mathrm{OPT}^{\mathrm{redir}}(\sigma):=
\min_{\mathcal D,\{m_i\}}
\left\{C|\mathcal D|+\sum_{i=1}^{L}
\mathrm{cost}^{\mathrm{redir}}_{\mathcal D,\{m_i\}}(\sigma_i)\right\}.
\end{equation}

\subsection{\textsc{CAPSUM-E+}: Redirect-Aware Service}
\label{subsec:capsum-e-plus}

\textsc{CAPSUM-E+} retains the local deployment decision of Algorithm~\ref{alg:capsum-e}. If this decision does not deploy an inactive local service, it queries a service directory and redirects to the nearest verified active instance when one exists; otherwise it uses the origin. This ordering is intentional: it makes the local deployment schedule identical to that of \textsc{CAPSUM-E}, enabling a direct comparison while still exploiting redirection whenever the local policy would pay the origin cost.

\begin{algorithm}[t]
\caption{\textsc{CAPSUM-E+} at node $n_i$ upon request $\sigma_i$}
\label{alg:capsum-e-plus}
\KwIn{$\sigma_i=(t_i,n_i,s_i,w_i)$; local predictor; service directory}
Update the local history of $(n_i,s_i)$ and set
$H\gets c_{n_i,s_i}(\sigma;(t_i-T,t_i])$\;
\uIf{$s_i$ is active locally}{
    Serve locally at cost $\beta w_i$\;
}
\Else{
    Query $\hat F\gets\hat c_{n_i,s_i}(\sigma;[t_i,t_i+T))$\;
    \uIf{$H\geq\gamma$ \textbf{and} $\hat F\geq\gamma$}{
        Deploy $s_i$ locally for $T$; pay $C$; serve locally\;
    }
    \Else{
        $\mathcal M\gets\{m\in\mathcal N:s_i\text{ is verified active at }m\}$\;
        \uIf{$\mathcal M\neq\emptyset$}{
            $m^*\gets\arg\min_{m\in\mathcal M}d(n_i,m)$; redirect to $m^*$\;
        }
        \Else{Forward to the origin\;}
    }
}
\end{algorithm}

Its cost is $C|\mathcal D_{\textsc{CAPSUM-E+}}(\sigma)|$ plus the request costs in \eqref{eq:redirect-cost} selected by the algorithm.

\subsection{Competitive Analysis}
\label{subsec:redir-analysis}

\begin{lemma}[Service-cost dominance]
\label{lem:service-dominance}
For every request sequence, \textsc{CAPSUM-E+} has the same local deployment schedule as \textsc{CAPSUM-E} and
\begin{equation}
\textsc{CAPSUM-E+}(\sigma)\leq\textsc{CAPSUM-E}(\sigma).
\label{eq:service-dominance}
\end{equation}
\end{lemma}

\begin{proof}
By induction over request arrivals, the two algorithms have the same local history and the same active state for every pair. They therefore make the same deployment decision at every inactive local request. If that decision does not deploy, \textsc{CAPSUM-E} uses the origin, whereas \textsc{CAPSUM-E+} uses either the origin or a verified remote instance whose cost is no greater by the routing-overhead assumption. Deployment costs are identical.
\end{proof}

\begin{lemma}[Relation between offline benchmarks]
\label{lem:projection}
For $\beta>0$,
\begin{equation}
\mathrm{OPT}^{\mathrm{loc}}(\sigma)
\leq\frac{1}{\beta}\,
\mathrm{OPT}^{\mathrm{redir}}(\sigma).
\label{eq:projection}
\end{equation}
\end{lemma}

\begin{proof}
Take any feasible redirectable schedule and retain its deployments. Replace every redirected request by origin service, retaining local and origin assignments. For a redirected request, $w_i\leq(\beta+\alpha d(n_i,m_i))w_i/\beta$; for local and origin requests the same factor is also an upper bound. Deployment costs obey $C\leq C/\beta$. The constructed local schedule therefore costs at most $1/\beta$ times the redirectable schedule. Minimizing both sides proves the claim.
\end{proof}

\begin{theorem}[Redirectable-model guarantee]
\label{thm:capsum-e-plus}
For $\beta>0$ and $\eta=\eta_{\textsc{CAPSUM-E+}}(\sigma)$,
\begin{equation}
\label{eq:capsum-e-plus-cr}
\frac{\textsc{CAPSUM-E+}(\sigma)}
     {\mathrm{OPT}^{\mathrm{redir}}(\sigma)}
\leq\frac{\kappa(\eta)}{\beta}.
\end{equation}
Consequently, \textsc{CAPSUM-E+} is $2/(\beta(1+\beta))$-consistent and $1/\beta^2$-robust with respect to the redirectable offline benchmark.
\end{theorem}

\begin{proof}
The decision points, and hence the prediction error, are the same for \textsc{CAPSUM-E+} and \textsc{CAPSUM-E}. Combine Lemma~\ref{lem:service-dominance}, Theorem~\ref{thm:main}, and Lemma~\ref{lem:projection}.
\end{proof}

\subsection{Practical Considerations}
\label{subsec:practical}

The guarantee treats directory lookup and control traffic separately from request-serving cost; their latency and bandwidth should be measured in the prototype evaluation. A stale directory entry must be verified before redirecting, with origin fallback on failure, to retain the service-cost dominance argument. A coordinator that suppresses a locally justified deployment can be a useful engineering heuristic for reducing replicas, but it changes the PFSUM schedule. Such a variant should be reported as a separate heuristic and should not be claimed to satisfy Theorem~\ref{thm:capsum-e-plus} without a new proof.
\newcommand{\ExperimentRuns}{895}
\newcommand{\CAPSUMBurstyCost}{0.765}
\newcommand{\CAPSUMBestGain}{42.9\%}
\newcommand{\CAPSUMEExactCost}{0.785}
\newcommand{\CAPSUMEBadCost}{0.808}
\newcommand{\FSUMBadCost}{0.940}
\newcommand{\RobustGain}{14.0\%}
\newcommand{\OverBiasGain}{13.2\%}
\newcommand{\PeakBurstGain}{39.9\%}
\newcommand{\GlobusTasks}{13,543}
\newcommand{\GlobusNodes}{20}
\newcommand{\GlobusServices}{217}
\newcommand{\GlobusDays}{212.9}
\newcommand{\GlobusCAPSUMCost}{0.725}
\newcommand{\GlobusCAPSUMGain}{45.5\%}

\section{Experimental Results}
\label{sec:evaluation}

We evaluate \textsc{CAPSUM-E}, the analyzable elastic specialization, separately
from finite-capacity \textsc{CAPSUM}, so capacity results are not presented as
evidence for the PFSUM guarantee. The same reproducible pipeline generates
every table and vector figure from retained per-seed results.

\subsection{Questions, Workloads, and Metrics}

Q1--Q2 test whether the elastic guard tracks the local offline optimum and
remains robust to prediction error. Q3--Q4 compare finite-capacity
\textsc{CAPSUM} with audited baselines and explain its cost anatomy. Q5--Q8
examine sensitivity, component necessity, capacity--quality interactions, and
scalability, while Q9 tests transfer to a public FaaS trace without forecast
leakage. The redirect extension is an analytical schedule-preservation result;
all finite-capacity runs exercise the same common redirect overlay and report
its remote-serving cost.

Synthetic traces contain 72 hours of ordered requests over a $9$-node grid and
60 services, with $w\in[70,130]$, $a_s=1+(s\bmod3)$, and stationary, regional
flash-crowd, or moving-hotspot demand. Background demand is Zipf-distributed
(default exponent 1.12); each flash-crowd seed has 10,464 requests. Defaults
are $T=6$ hours, $C=1000$, $\beta=0.25$, $\alpha=0.12$, and $B_n=9$.

The main, prediction, bias, and ablation studies use five seeds
$\{7,17,29,43,61\}$; sensitivity and interaction points use three. The
\ExperimentRuns{} policy runs span theory, main comparisons, prediction error
and bias, ablation, parameter and economic stress, public trace, and one- and
two-dimensional scaling; the manifest gives exact per-study counts. Error bars are two-sided 95\%
Student-$t$ intervals across seeds (or five chronological trace blocks), and paired
policies receive identical event streams and seeded random choices wherever
the algorithms permit. Thus paired gaps suppress trace-generation variance;
the intervals quantify run-to-run variation rather than population sampling
uncertainty. Trace-block intervals are descriptive because adjacent
chronological blocks need not be independent.

The primary metric is total deployment, local, remote, and origin cost divided
by origin-only cost; Proposition~\ref{prop:scale-invariance} shows that this
normalization is independent of the workload-cost unit. We also report edge-service rate, cost composition,
deployments, evictions, empirical local ratio, and decision-path time.
In every finite-capacity run, a local miss is redirected to the nearest active
copy when one exists and otherwise goes to the origin; this common serving
overlay does not alter any policy's placement decisions. The configured
$\alpha$ satisfies the routing-cost condition in Section~\ref{sec:redirectable}.

The synthetic studies use the controlled predictor
\begin{equation}
 \hat F=\max\{0,F+(2U-1)(1-q)\gamma\},\qquad U\sim\mathrm{Uniform}(0,1),
\end{equation}
using a deterministic event-indexed draw. Thus $q=1$ is exact and
$\eta/\gamma$ is controlled in the workload units of
Theorem~\ref{thm:main}; $q\in\{1,.8,.6,.4,.2,0\}$. Because this predictor is
constructed from the realized future $F$, it is an oracle-style mechanism for
controlled error injection, not a deployable forecaster. The public-trace
study instead uses the strictly causal predictor described below.

\subsection{Implementation and Reproducibility}

The dependency-light Python implementation processes every trace in stable
timestamp order and uses explicit tie breaking for deterministic policy
decisions. For the elastic theory track, it computes an exact offline optimum
independently for every node--service stream. Let $j(i)$ be the first request
at or after $t_i+T$ and let $D_i$ denote optimal suffix cost. A two-pointer
scan and suffix workload sums evaluate
\begin{equation}
 D_i=\min\!\left\{w_i+D_{i+1},
 C+\beta\!\sum_{k=i}^{j(i)-1}w_k+D_{j(i)}\right\},
 \label{eq:offline-dp}
\end{equation}
with $D_{L+1}=0$. The request-time deployment property in
Section~\ref{sec:problem} makes this dynamic program exact; it is used only for
the local uncapacitated ratio and never as a finite-capacity oracle.

One driver regenerates the manifest, per-seed rows, tables, and figures; an
audit checks counts, cost decomposition, ranges, empirical PFSUM-bound
conformance, and the trace checksum. Tests cover
TTL boundaries, the dual gate, the dynamic program against brute force,
capacity, renewal, prediction, and adapter paths. Decision-path time brackets policy execution
with a monotonic nanosecond clock; trace generation, the offline oracle, file
I/O, and figure construction are excluded. Consequently, the timing results
compare controller computation, not network, transfer, or cold-start latency.
Each policy--workload run is timed once; seed variation is not a microbenchmark
repeat. Runs used CPython 3.12.13 on a 15-core Apple M5 Pro with 48~GB memory.

\subsection{Baseline Fidelity}

We compare directly against SUM, FSUM, LRU, LFU, and Random in the common TTL
model. We also include source-derived adapters motivated by
EDP-A~\cite{luo2025placement}, OREO~\cite{xu2018joint}, and
uEDC-L~\cite{luo2025cost}, representing popularity-based placement,
energy-aware joint caching/offloading, and failure-robust caching. Their native
decision objects and objectives differ from executable TTL-bound services, so
an exact reproduction cannot simultaneously preserve the source problem and
produce actions in our model. Table~\ref{tab:baseline-fidelity} therefore states
every heuristic translation. These adapters are auditable comparison points,
not native reproductions of the three source systems.

\begin{table*}[t]
\centering\scriptsize
\caption{Direct baselines and explicitly scoped source-derived adapters.}
\label{tab:baseline-fidelity}
\setlength{\tabcolsep}{3.5pt}
\begin{tabularx}{\textwidth}{@{}p{0.10\textwidth}p{0.27\textwidth}p{0.34\textwidth}X@{}}
\toprule
\textbf{Policy} & \textbf{Source concept used} & \textbf{Common-model implementation} & \textbf{Unpreserved native semantics} \\
\midrule
EDP-A~\cite{luo2025placement} & Popularity order and proximity coverage. & Greedily cover prior-epoch demand within 1.01 grid hops. & No erasure code, $K$-block recovery, or remaining-block pass. \\
OREO~\cite{xu2018joint} & Slot observation, virtual queue, drift scoring, and Gibbs proposals. & One-hour slots; one queue; 80 proposals; $V=4$; temperature .08; target .58. & Replaces native offloading, radio/computation latency, and energy. \\
uEDC-L~\cite{luo2025cost} & Popularity/failure uncertainty and redundancy. & Score $d+.35|d-d'|+.25d_{\rm nbr}$; 0/1 knapsack; two-copy repair. & Not the source robust program or continuous retrieval recourse. \\
\bottomrule
\end{tabularx}
\end{table*}

The native objectives are not directly comparable. Each adapter instead
produces a deployment and routing sequence, to which the common accounting
layer applies the same $a_sC$, $\beta w$, remote, and origin costs. A placement
is free only until its paid TTL expires; continuing it after expiry incurs
another $a_sC$, and early eviction yields no refund. The OREO-derived adapter
observes a complete current slot before acting, matching its source information
boundary but giving it more within-slot information than the causal epoch
adapters. Epoch-based placement uses only completed past epochs. Every
pseudorandom choice is seeded. Direct-baseline definitions and full adapter
pseudocode appear in the artifact; adapter parameters are fixed across traces.

\begin{table*}[t]
\centering\scriptsize
\caption{Finite-capacity common-model comparison (mean $\pm$ 95\% Student-$t$ interval over five seeds). Cost is normalized by origin-only service; lower is better. Best value in each row is bold.}\label{tab:main-capacity}
\setlength{\tabcolsep}{3.0pt}
\begin{tabular}{llrrrrrrr}
\toprule \textbf{Workload} & \textbf{Metric} & \textbf{CAPSUM} & \textbf{EDP-A} & \textbf{OREO} & \textbf{uEDC-L} & \textbf{LRU} & \textbf{LFU} & \textbf{Random} \\
\midrule
Stationary & Norm. cost & $\mathbf{0.895 \pm 0.005}$ & 1.365 $\pm$ 0.006 & 2.910 $\pm$ 0.039 & 1.350 $\pm$ 0.005 & 10.654 $\pm$ 0.281 & 8.795 $\pm$ 0.187 & 11.378 $\pm$ 0.240 \\
 & Edge rate (\%) & 33.6 $\pm$ 1.6 & 74.6 $\pm$ 0.7 & 67.3 $\pm$ 1.6 & 70.4 $\pm$ 0.7 & $\mathbf{100.0 \pm 0.0}$ & $\mathbf{100.0 \pm 0.0}$ & $\mathbf{100.0 \pm 0.0}$ \\
Flash crowd & Norm. cost & $\mathbf{0.765 \pm 0.005}$ & 1.340 $\pm$ 0.004 & 2.989 $\pm$ 0.058 & 1.344 $\pm$ 0.010 & 7.556 $\pm$ 0.107 & 7.353 $\pm$ 0.093 & 8.312 $\pm$ 0.160 \\
 & Edge rate (\%) & 49.7 $\pm$ 0.8 & 76.5 $\pm$ 0.4 & 62.7 $\pm$ 1.7 & 70.4 $\pm$ 1.7 & $\mathbf{100.0 \pm 0.0}$ & $\mathbf{100.0 \pm 0.0}$ & $\mathbf{100.0 \pm 0.0}$ \\
Moving hotspot & Norm. cost & $\mathbf{0.726 \pm 0.004}$ & 1.233 $\pm$ 0.011 & 2.647 $\pm$ 0.083 & 1.248 $\pm$ 0.011 & 6.613 $\pm$ 0.109 & 6.608 $\pm$ 0.137 & 7.389 $\pm$ 0.131 \\
 & Edge rate (\%) & 56.8 $\pm$ 0.8 & 75.3 $\pm$ 1.5 & 64.9 $\pm$ 2.3 & 67.5 $\pm$ 1.8 & $\mathbf{100.0 \pm 0.0}$ & $\mathbf{100.0 \pm 0.0}$ & $\mathbf{100.0 \pm 0.0}$ \\
\bottomrule\end{tabular}\end{table*}

\begin{table}[t]
\centering\scriptsize
\caption{Elastic CAPSUM-E validation at $q=0.7$ (mean $\pm$ 95\% Student-$t$ interval; five seeds). Best value in each column is bold.}\label{tab:theory-track}
\begin{tabular}{lrr}
\toprule \textbf{Policy} & \textbf{Norm. cost} & $\mathbf{\mathcal{A}/\mathrm{OPT}^{\mathrm{loc}}}$ \\
\midrule
CAPSUM-E & 0.789 $\pm$ 0.005 & 1.078 $\pm$ 0.004 \\
SUM & 0.821 $\pm$ 0.006 & 1.122 $\pm$ 0.005 \\
FSUM & $\mathbf{0.756 \pm 0.002}$ & $\mathbf{1.034 \pm 0.004}$ \\
No deployment & 1.000 $\pm$ 0.000 & 1.367 $\pm$ 0.007 \\
\bottomrule\end{tabular}\end{table}

\begin{table}[t]
\centering\scriptsize
\caption{CAPSUM ablation (mean normalized cost; five seeds). Minimum in each column is bold; deployment count is diagnostic.\label{tab:capsum-ablation}}
\setlength{\tabcolsep}{3.2pt}\begin{tabular}{lrrrr}
\toprule \textbf{Variant} & \textbf{Steady} & \textbf{Flash} & \textbf{Mobile} & \textbf{Flash deploys} \\
\midrule
Full & 0.895 & 0.765 & 0.726 & \textbf{98.4} \\
History only & 0.906 & 0.796 & 0.751 & 147.0 \\
Forecast only & \textbf{0.891} & 0.766 & \textbf{0.715} & 161.6 \\
No shadow price & 0.896 & \textbf{0.762} & 0.722 & 101.0 \\
\bottomrule\end{tabular}\end{table}

\begin{table}[t]
\centering\scriptsize
\caption{Globus results over five post-burn-in chronological blocks. Student-$t$ intervals describe temporal variability, not sampling uncertainty.}\label{tab:globus-trace}
\setlength{\tabcolsep}{4.2pt}\begin{tabular}{lrr}
\toprule \textbf{Policy} & \textbf{Norm. cost} & \textbf{Edge rate (\%)} \\
\midrule
CAPSUM & 0.725 $\pm$ 0.176 & 53.6 $\pm$ 16.8 \\
EDP-A & 1.332 $\pm$ 0.240 & 16.8 $\pm$ 20.5 \\
OREO & 41.380 $\pm$ 20.394 & 36.7 $\pm$ 16.2 \\
uEDC-L & 7.626 $\pm$ 3.195 & 16.2 $\pm$ 19.5 \\
\bottomrule\end{tabular}\end{table}

\begin{figure*}[t]
\centering
\includegraphics[width=0.94\textwidth]{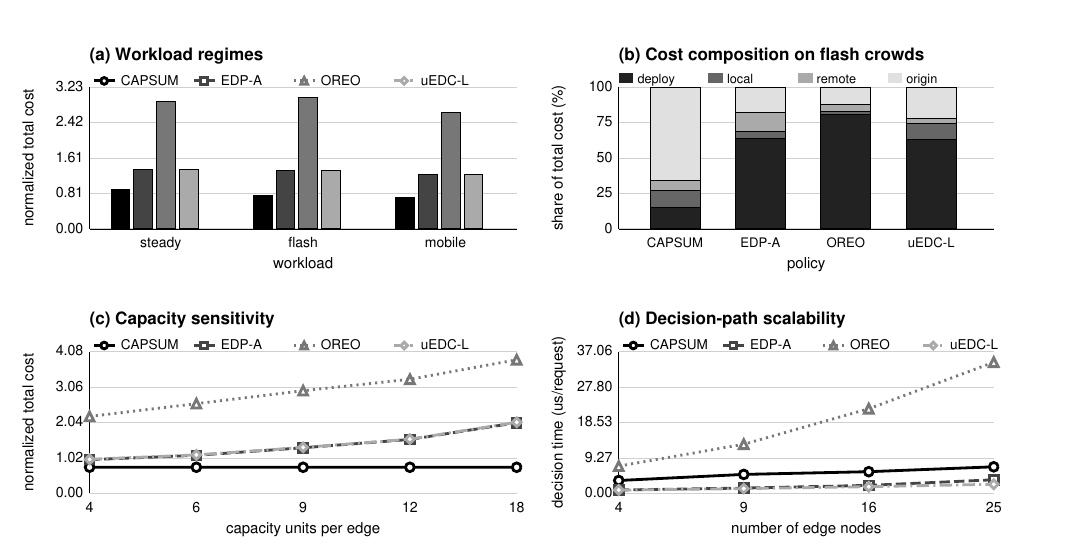}
\caption{Summary: (a) demand regimes, (b) flash-crowd cost composition,
(c) capacity sensitivity, and (d) decision-path scaling. Curves show seed means.}
\label{fig:evaluation-overview}
\end{figure*}

\begin{figure*}[t]
\centering
\includegraphics[width=0.94\textwidth]{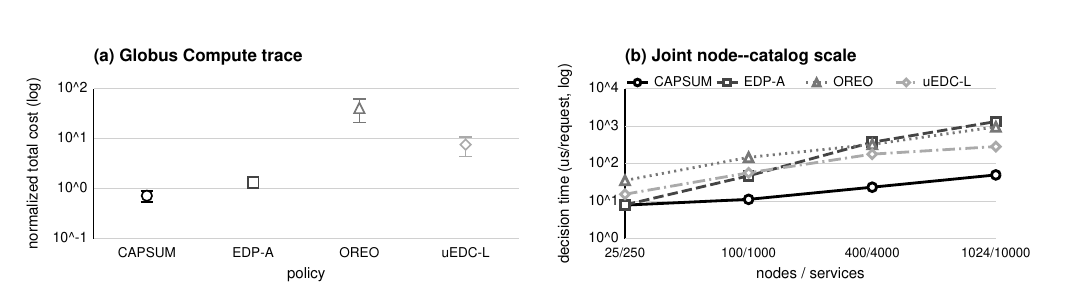}
\caption{Trace and joint scaling: (a) Globus block mean and descriptive 95\%
Student-$t$ interval; (b) decision time with about 24,000 requests/configuration
and three seeds. Patterns and markers identify methods.}
\label{fig:trace-scale}
\end{figure*}

\begin{figure*}[t]
\centering
\includegraphics[width=0.91\textwidth]{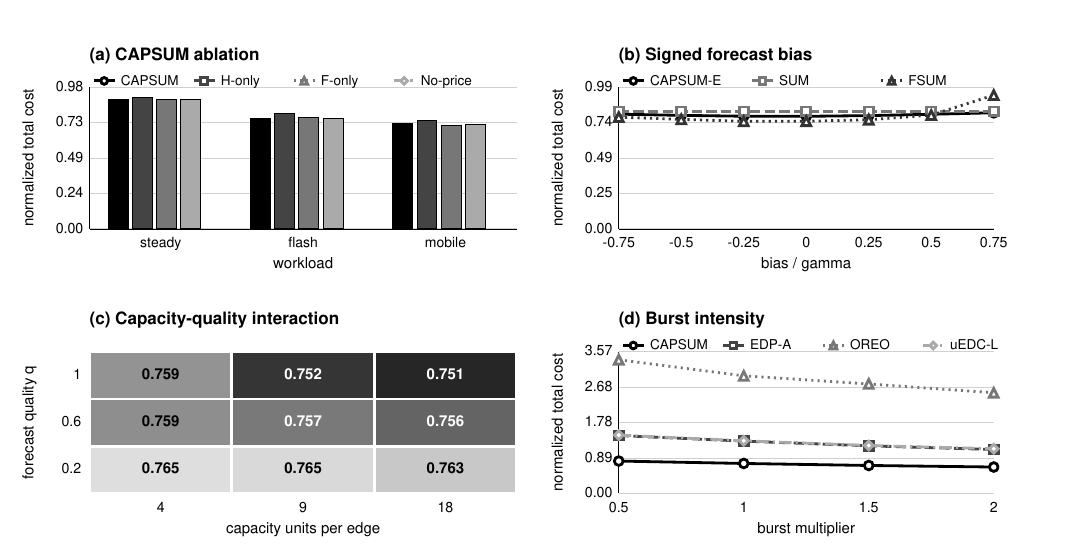}
\caption{Robustness diagnostics: (a) ablation, (b) signed bias,
(c) capacity--quality interaction, and (d) burst stress. Bars/curves are seed means.}
\label{fig:robustness-diagnostics}
\end{figure*}

\begin{figure*}[t]
\centering
\includegraphics[width=0.89\textwidth]{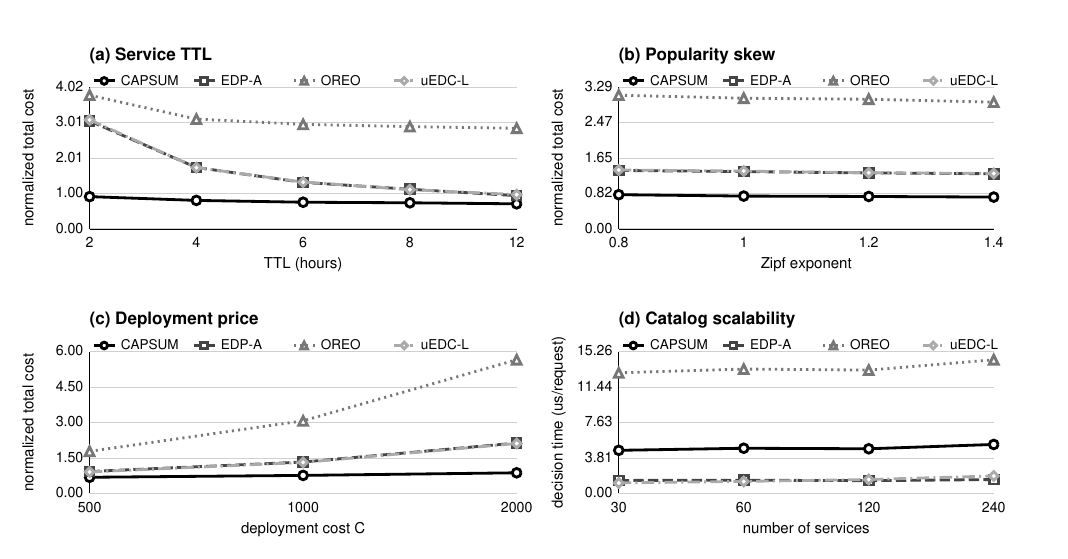}
\caption{Sensitivity to (a) TTL, (b) Zipf exponent, (c) deployment price at
$\beta=.25$, and (d) catalog size. Points are three-seed means.}
\label{fig:extended-sensitivity}
\end{figure*}

\subsection{Elastic Safety Layer}

At $q=0.7$, Table~\ref{tab:theory-track} gives \textsc{CAPSUM-E} cost
$0.789\pm0.005$ and empirical local-optimum ratio $1.078\pm0.004$ (Q1).
SUM is more conservative ($0.821\pm0.006$); FSUM is lower under this moderately
accurate advice ($0.756\pm0.002$), as expected from its consistency bias.

With exact advice, Figure~\ref{fig:prediction-sensitivity} gives
\textsc{CAPSUM-E} cost \CAPSUMEExactCost; near $\eta/\gamma=1$, FSUM rises to
\FSUMBadCost\ versus \CAPSUMEBadCost\ for \textsc{CAPSUM-E}, a \RobustGain\ reduction
(Q2). SUM is prediction independent. The crossing exposes the expected
consistency--robustness trade-off.

Signed-error results in Figure~\ref{fig:robustness-diagnostics}(b) show that
underprediction mainly delays both policies, whereas positive bias creates
unsupported FSUM purchases: at $0.75\gamma$, \textsc{CAPSUM-E} costs $0.810$
versus $0.933$ (\OverBiasGain\ lower), matching
Proposition~\ref{prop:history-guard}.

\begin{figure}[t]
\centering
\includegraphics[width=0.90\columnwidth]{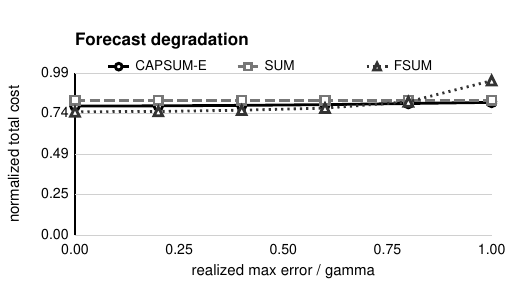}
\caption{Prediction stress test over six controlled error levels. The x-axis
is realized workload-valued error, not a classifier accuracy score.}
\label{fig:prediction-sensitivity}
\end{figure}

\subsection{Finite-Capacity Comparison and Cost Anatomy}

Within the common TTL model, \textsc{CAPSUM} has the lowest cost in all three
regimes (Table~\ref{tab:main-capacity}, Q3): $0.895\pm0.005$ for stationary
demand, \CAPSUMBurstyCost{} for flash crowds, and $0.726\pm0.004$ for moving
hotspots. Its flash-crowd gain over the next-best source-derived adapter is
\CAPSUMBestGain; the paired gap to the EDP-A-derived adapter is
$0.575\pm0.005$, with all five seeds favoring \textsc{CAPSUM}. This ranking is
for the adapters in Table~\ref{tab:baseline-fidelity}, not the source systems'
native objectives.

Edge-service rate alone is insufficient: on flash crowds EDP-A serves $76.5\%$
at an edge but costs $1.340\pm0.004$, OREO serves $62.7\%$ but costs
$2.989\pm0.058$, and immediate-admission LRU/LFU/Random obtain $100\%$ only
with costs 7.35--8.31. Common accounting exposes repeated instantiations that
their native objectives do not minimize.

Figure~\ref{fig:evaluation-overview}(a)--(b) attributes these outcomes (Q4):
OREO is deployment-dominated, while EDP-A and uEDC-L trade more deployment for
edge coverage; \textsc{CAPSUM} delays admission until workload amortizes both
image cost and current capacity price.
The history-only \textsc{CAPSUM-H} ablation also remains below every
source-derived adapter in all three synthetic regimes
(Table~\ref{tab:capsum-ablation}), so the ordering is not solely an artifact of
the controlled future oracle.

\subsection{Public-Trace Validation}

For Q9, we convert the 546.8-MB Globus Compute task table, whose
endpoint--function structure best matches $(t,n,s,w)$~\cite{bauer2024globus}.
The script samples 24 uniformly spaced 512-KiB ranges, removes boundary
fragments and incomplete records, deduplicates task identifiers, parses 39,764
unique tasks, and retains valid durations among the 32 most active endpoints
and 256 functions. Filtering after endpoint/function selection deliberately
preserves chronology and can yield fewer than the nominal maxima. The result
has \GlobusTasks{} requests, \GlobusNodes{} endpoints, and \GlobusServices{}
functions over \GlobusDays{} days.

Recorded submissions determine arrival time; fifth/95th-percentile-winsorized
duration maps monotonically to $w\in[70,130]$, and median argument size maps by
rank tercile to $a_s\in\{1,2,3\}$. One pair of records shares a nanosecond
timestamp; event identifier breaks the tie, and the latter record is advanced
to the next representable floating-point time to satisfy the strictly ordered
model without changing trace-scale timing. We use $T=7$ days and $B_n=24$. A causal
predictor extrapolates only the preceding 28 days of pair-local workload. After
one-sixth burn-in, five contiguous equal-count blocks are measured. Each block
is preceded by up to 28 days of state warm-up, after which accounting counters
are reset while the warmed predictor and cache state are retained. Predictions
at an arrival use only records already observed at that timestamp; neither a
future block nor a future record in the current block enters the forecast.

\textsc{CAPSUM} costs \GlobusCAPSUMCost{} and serves $53.6\%$ at an edge
(Table~\ref{tab:globus-trace}, Fig.~\ref{fig:trace-scale}(a)),
\GlobusCAPSUMGain{} below EDP-A's $1.332$. OREO costs $41.380$ at $36.7\%$
edge service and installs 4,104.0 images/block; uEDC-L costs $7.626$ and
installs 574.0, versus 13.6 for \textsc{CAPSUM} and 37.4 for EDP-A. Thus edge
coverage helps only when deployment cost is amortized.

\subsection{Component Ablation and Cross-Factor Robustness}

For Q7, removing the forecast guard raises flash-crowd cost from $0.765$ to
$0.796$ and mobility cost from $0.726$ to $0.751$
(Table~\ref{tab:capsum-ablation}, Fig.~\ref{fig:robustness-diagnostics}(a)).
Forecast-only and no-price variants are slightly cheaper on some accurate
traces, quantifying the price of pathwise safety and scarcity awareness; panel
(b) shows that the former saving fails under positive bias.

For Q8, expanding capacity from 4 to 18 changes cost by only $0.002$ at $q=.2$
but lowers it from $0.759$ to $0.751$ at $q=1$ (panel (c)). At burst multiplier
2, \textsc{CAPSUM} costs $0.665$ and beats the best source-derived adapter by
\PeakBurstGain\ (panel (d)), excluding a single-amplitude artifact.

\subsection{Capacity, TTL, Popularity, and Scale}

Figure~\ref{fig:evaluation-overview}(c) and
Figure~\ref{fig:extended-sensitivity}(a)--(c) answer Q5. Extra capacity can
increase EDP-A, OREO, and uEDC-L deployment cost, whereas \textsc{CAPSUM}'s
shadow price keeps its curve nearly flat. Longer TTL amortizes deployment
(largest gain from two to four hours); greater Zipf skew concentrates workload
without changing the flash-crowd ranking. When $C$ rises from 500 to 2000,
\textsc{CAPSUM} increases from $0.684$ to $0.873$, versus $2.134$, $5.662$, and
$2.116$ for EDP-A, OREO, and uEDC-L, respectively.

Figures~\ref{fig:evaluation-overview}(d) and
\ref{fig:extended-sensitivity}(d) answer Q6. Local evidence and knapsack paths
remain below $7\,\mu$s/request through 25 nodes; OREO reaches about
$34\,\mu$s because each Gibbs proposal reevaluates global offloading. Over
30--240 services, \textsc{CAPSUM} changes from $4.64$ to $5.27\,\mu$s, uEDC-L
from $1.13$ to $1.88\,\mu$s, and OREO from $12.99$ to $14.39\,\mu$s. These are
simulator decision times, not network or cold-start latency.

The joint sweep (Fig.~\ref{fig:trace-scale}(b)) uses 25/250, 100/1,000,
400/4,000, and 1,024/10,000 node/service pairs, about 24,000 requests, three
seeds, one-hour epochs, and $B_n=24$; $\alpha$ is diameter-normalized and
zero-demand EDP-A/uEDC-L entries are omitted. At the largest scale, mean times
are $50.38$, $287.43$, $962.12$, and $1{,}334.10\,\mu$s/request for
\textsc{CAPSUM}, uEDC-L, OREO, and EDP-A, exposing global-placement overhead
absent from one-dimensional sweeps.

\subsection{Validity Boundary}

The rankings are scoped. Globus removes the synthetic-arrival limitation but
uses a byte-stratified sample with proxy costs and image sizes. Azure lacks
edge locations~\cite{shahrad2020serverless}, while Alibaba has a shorter,
cluster-centric horizon~\cite{luo2021alibaba}. Simulator time omits prediction,
lookup, transfer, cold start, and redirection; common accounting omits native
reliability, energy, and latency objectives. Thus the results characterize the
stated TTL model rather than the source systems universally.

\section{Conclusion}
\label{sec:conclusion}

\textsc{CAPSUM-E} decomposes into node-service Bahncard instances and inherits
PFSUM's tight $\kappa(\eta)$ trade-off; its redirect overlay preserves the local
schedule and is $\kappa(\eta)/\beta$-competitive against the redirectable
benchmark. Finite-capacity \textsc{CAPSUM} adds shadow-priced evidence and
density eviction with feasibility and scale invariance, but no claimed capacity
ratio. Exact-offline and prediction-stress tests target the elastic mechanism;
ablation, causal-trace, and scale tests target finite admission. Across 895
common-model runs, \textsc{CAPSUM} has the lowest cost in all three synthetic
regimes and the sampled Globus trace. Future work should couple capacity and
routing in the offline benchmark and measure transfer, cold-start, prediction,
and directory costs in a container deployment.

\bibliographystyle{IEEEtran}
\bibliography{ref}








\end{document}